\documentclass[conference,a4paper]{IEEEtran}

\usepackage{amsfonts}
\usepackage{amssymb}
\usepackage{amsmath}
\usepackage{bbm}
\usepackage{bm}
\usepackage{dsfont}
\usepackage{graphicx}

\newcommand{\braket}[2]{#1^\dagger #2}

\newcommand{\myf}{f}
\newcommand{\myF}{\bm{f}}
\newcommand{\mypsi}{\psi}
\newcommand{\myPsi}{\bm{\psi}}

\newcommand{\myT}{\mathcal{T}}

\newcommand{\Pe}{\mathsf{P}_{\text{e}}}

\newtheorem{theorem}{Theorem}
\newtheorem{lemma}{Lemma}
\newtheorem{remark}{Remark}

\begin{document}

\sloppy

%% Paper Title
%% You can use linebreaks \\ within to get better formatting as
%% desired. 
\title{An Elias Bound on the Bhattacharyya Distance of Codes for Channels with a Zero-Error Capacity}

%% Author names and affiliations:
%%
%% Avoiding spaces at the end of the author lines is not a problem with
%% conference papers because we don't use \thanks or \IEEEmembership.
%%
%% For several authors with only one affiliation:
%%
% \author{
%   \IEEEauthorblockN{Hui-Ting Chang and Stefan M.~Moser}
%   \IEEEauthorblockA{Department of Electrical and Computer Engineering\\
%     National Chiao Tung University (NCTU)\\
%     Hsinchu, Taiwan\\
%     Email: \{email-of-hui-ting,email-of-stefan\}@ieee.org} 
% }
%%
%% For up to three affiliations:
%%
\author{
  \IEEEauthorblockN{Marco Dalai}
  \IEEEauthorblockA{Department of Information Engineering\\
    University of Brescia - Italy\\
    Email: marco.dalai@unibs.it} 
}
%%
%% For over three affiliations, or if they all won't fit within the width
%% of the page, use this alternative format:
%%
% \author{
%   \IEEEauthorblockN{
%     Michael Shell\IEEEauthorrefmark{1},
%     Homer Simpson\IEEEauthorrefmark{2},
%     James Kirk\IEEEauthorrefmark{3}, 
%     Montgomery Scott\IEEEauthorrefmark{3} and
%     Eldon Tyrell\IEEEauthorrefmark{4}}
%   \IEEEauthorblockA{
%     \IEEEauthorrefmark{1}School of Electrical and Computer Engineering\\
%     Georgia Institute of Technology, Atlanta, Georgia 30332--0250\\ 
%     Email: see http://www.michaelshell.org/contact.html}
%   \IEEEauthorblockA{
%     \IEEEauthorrefmark{2}Twentieth Century Fox, Springfield, USA\\
%     Email: homer@thesimpsons.com}
%   \IEEEauthorblockA{
%     \IEEEauthorrefmark{3}Starfleet Academy, San Francisco, California 96678-2391\\
%     Telephone: (800) 555--1212, Fax: (888) 555--1212}
%   \IEEEauthorblockA{
%     \IEEEauthorrefmark{4}Tyrell Inc., 123 Replicant Street, Los Angeles, California 90210--4321}
% }

%% Use for special paper notices

%% To balance the two columns, you should reduce the text-height of
%% the last page using the following command:
%%%%%%%%%%%%%%%%%%%%%%%%%%%%%%%%%%%%%%%%%%%%%%%%%%%%%%%%%%%%%%%%%%%%%
%\addtolength{\textheight}{-9.35cm}
%%%%%%%%%%%%%%%%%%%%%%%%%%%%%%%%%%%%%%%%%%%%%%%%%%%%%%%%%%%%%%%%%%%%%
%% with an appropriate value. This command must be place on the second
%% last page, i.e., for a one-page abstract here, for a two-page
%% abstract right after the \maketitle command.

%% Create the title:
\maketitle

%% Abstract: 
%% For the final version of the accepted paper, please make sure you
%% remove the comment "THIS PAPER IS ELIGIBLE FOR THE STUDENT PAPER
%% AWARD."
%%
\begin{abstract}
In this paper, we propose an upper bound on the minimum Bhattacharyya distance of codes for channels with a zero-error capacity.
The bound is obtained by combining an extension of the Elias bound  introduced by  Blahut, with an extension of a bound previously introduced by the author, which builds upon ideas of Gallager, Lov\'asz and Marton. 

\end{abstract}

\section{Introduction}
\label{sec:Intro}
An intriguing problem in the study of discrete memoryless channels (DMC) is that of determining the asymptotic behavior of the probability of error $\Pe$ of optimal codes in the low rate region. In the most general case, the probability of error is precisely zero at rates below the so called zero-error capacity $C_0$, while for $R>C_0$ it is known to be an exponential function of the block-length $n$, i.e.
\begin{equation*}
\Pe\approx e^{-n E(R)},
\end{equation*}
where $E(R)$ is the so called reliability function of the channel. Both $C_0$ and $E(R)$ in the proximity of $C_0$ are unknown in the general case. The most effective upper bounds to $C_0$ and to $E(R)$ were developed independently and there is not yet a good unified and consistent upper bound to both quantities.

In recent works by this author, a possible approach for unifying bounds to $C_0$ and $E(R)$ was suggested which attempts at bounding the Bhattacharyya minimum distance of codes at rates $R>\vartheta$, where $\vartheta$ is Lov\'asz' upper bound to $C_0$ \cite{lovasz-1979}. However, the bounds derived in \cite{dalai-TIT-2013, dalai-ISIT-2013b} are rather crude and there seems to be room for great improvements. For example, when used for the binary symmetric channel, the bound in \cite{dalai-ISIT-2013b} gives essentially the simple Plotkin bound for the zero-rate minimum distance of codes. A useful progress with respect to \cite{dalai-ISIT-2013b} would be a refinement of the ideas to obtain a bound which is both valid in the case of zero-error capacity and not as bad in the case of no zero-error capacity.

In this paper, we make a first step in this direction by proposing an evolution of the idea presented in \cite{dalai-ISIT-2013b} to bound the Bhattacharyya minimum distance of codes on channels with a zero-error capacity. The obtained bound can be interpreted as an extension of the Elias bound to this setting, and it is based on a combination of ideas introduced by Elias, Blahut \cite{blahut-1977}, Gallager \cite{gallager-1965}, Lov\'asz \cite{lovasz-1979} and Marton \cite{marton-1993}.

\section{Elias Bounds}
Generalizations of the Elias bound to non-binary channels have already appeared in the literature. The main contributions in this direction are those of Berlekamp \cite[Ch. 13]{berlekamp-book-1984}, Blahut \cite{blahut-1977}, and Piret \cite{piret-1986}. All those extensions are based on some notion of distance $d(x,x')$ between symbols $x$, $x'$, and distance $d(\bm{x},\bm{x}')$ between codewords $\bm{x}$ and $\bm{x}'$, and follow a scheme based on two steps. For a given code, one first identifies a subset $\myT$ of codewords which are all \emph{packed} in a ball around a properly chosen fixed sequence $\bar{\bm{x}}$. Then, the Plotkin bound is used to bound the minimum distance of the code in terms of the average distance between pairs of distinct codewords in $\myT$ as
\begin{equation}
d_{\text{min}} \leq \frac{1}{|\myT|(|\myT|-1)}\sum_{\bm{x},\bm{x}'\in \myT} d(\bm{x},\bm{x}').
\label{eq:plotkin}
\end{equation}
An important point in this scheme is that the distance used for sequences must be based on the additive application of the distance $d(x,x')$ between symbols, which means that for sequences $\bm{x}=(x_1,\ldots,x_n)$ and $\bm{x}'=(x'_1,\ldots,x'_n)$ we have
\[
d(\bm{x},\bm{x}')=\sum_{i=1}^n d(x_i,x'_i).
\]
This allows one to rewrite the average in eq. \eqref{eq:plotkin} in terms of the componentwise distances as
\begin{align}
\sum_{\bm{x},\bm{x}'\in \myT} d(\bm{x},\bm{x}') & =\sum_{\bm{x},\bm{x}'\in \myT}\sum_i d(x_i,x'_i)\\
& = \sum_i \left(\sum_{\bm{x},\bm{x}'\in \myT}d(x_i,x'_i)\right).
\label{eq:component_plotkin}
\end{align}
Then, the constraints on the compositions of the sequences $\bm{x}, \bm{x}'$ (and possibly $\bar{\bm{x}}$) are used to derive the final bound in slightly different ways in the different approaches \cite{berlekamp-book-1984}, \cite{blahut-1977}, \cite{piret-1986}.

When one considers the case of general DMCs, the first problem is that of choosing a meaningful distance between sequences and symbols. If we are interested in understanding the probability of error of optimal codes, a reasonable approach is to consider the use of the Bhattacharyya distance between symbols\footnote{We point out that bounds on the Bhattacharyya distance of codes can be immediately mapped to bounds on the reliability function for certain symmetric channels, for example for pairwise reversible channels in the sense of \cite{shannon-gallager-berlekamp-1967-2} and (with obvious redefinitions of quantities) for classical-quantum pure-state channels. Due to space limitation, we leave the discussion of such applications of our bound to future work. 
}
\begin{equation}
d(x,x')=-\log \sum_y \sqrt{W(y|x)W(y|x')},
\end{equation}
where $W(y|x)$ is the channel transition probability from input $x$ to output $y$.
%For channels with certain symmetries, at least, good bounds to the minimum Bhattaharyya distance imply good bounds on the reliability function. 
Of the three mentioned generalizations, the only one which considers the case of general DMCs is that of Blahut, which actually studies the minimum Bhattacharyya distance of codes.

Blahut focuses on a subset of channels previously studied by Jelinek \cite{jelinek-1968} and he considers the case of no zero-error capacity. There is a strong technical reason for this choice. In fact, if the channel has a zero-error capacity, the Bhattacharyya distance $d(x,x')$ is infinite for some pairs of inputs $x$, $x'$. So, optimal codes will in general contain pairs of codewords $\bm{x},\bm{x}'$ with infinite distance, and any attempt to use the Plotking averaging procedure of equation \eqref{eq:plotkin} fails, since it gives the trivial bound $d_{\text{min}}\leq \infty$.

In this paper, we propose an extension of the Elias bound for channels with a zero-error capacity by considering a variation of the Plotkin step. In a nutshell, since infinite distances arise from the use of the logarithm, we get rid of the logarithm or, equivalently, rather than averaging the pairwise distances $d(\bm{x},\bm{x}')$, we average an exponential function of those distances.
In particular, we use an approach which in a sense corresponds to substituting equation \eqref{eq:plotkin} with 
\begin{equation}
d_{\text{min}} \leq -\rho \log\left( \max_{\bm{x}\in \myT}\frac{1}{(|\myT|-1)}\sum_{\bm{x}'\in \myT\backslash\{\bm{x}\}} e^{-d(\bm{x},\bm{x}')/\rho}\right).
\label{eq:exp_plotkin}
\end{equation}
There is a drawback of course, in that the derivation of the bound must now follow a different route, since it is no longer possible to use eq. \eqref{eq:component_plotkin}. We approach the problem by proposing an extension of the umbrella bound originally introduced in \cite{dalai-ISIT-2013b}. That bound can in fact be interpreted as a variation of the Plotkin bound \eqref{eq:plotkin} in the form of equation \eqref{eq:exp_plotkin}, when there is no constraint on the composition of the codewords $\bm{x}, \bm{x}'$. Here, we propose an extension of the method that allows us to handle composition constraints as is usually done with equation \eqref{eq:component_plotkin}.

In the next section, we introduce the notation and report the basic result of \cite{dalai-ISIT-2013b} for the reader's convenience. We then propose a way to deal with composition constraints and present the associated generalization of the Elias bound. We finally discuss how this bound relates to previously known ones.

\section{$\vartheta(\rho)$ and the basic Umbrella bound}

Let $W(y|x)$, $x\in \mathcal{X}$, $y\in\mathcal{Y}$, be the transition probabilities of a discrete memoryless channel $W$ with input alphabet 
% :\mathcal{X}\to\mathcal{Y}$, where $\mathcal{X}=\{1,2,\ldots,K\}$ and $\mathcal{Y}=\{1,2,\ldots,J\}$ are finite sets
$\mathcal{X}$ and output alphabet $\mathcal{Y}$. For a sequence $\bm{x}=(x_1,x_2,\ldots,x_n)\in\mathcal{X}^n$ and a sequence $\bm{y}=(y_1,y_2,\ldots,y_n)\in\mathcal{Y}^n$, the probability of observing $\bm{y}$ at the output of the channel given $\bm{x}$ at the input is 
\begin{equation}
W^{(n)}(\bm{y}|\bm{x})=\prod_{i=1}^n W(y_i|x_i).
\end{equation}

For a generic input symbol $x$, consider the unit norm $|\mathcal{Y}|$-dimensional column ``state'' vector $\mypsi_x$ with components $\mypsi_x(y)=\sqrt{W(y|x)}$. In the same way, for an input sequence $\bm{x}=(x_1,x_2,\ldots,x_n)$, consider the unit norm $|\mathcal{Y}|^n$-dimensional column vector $\myPsi_{\bm{x}}$ whose components are the values
$\sqrt{W^{(n)}(\bm{y}|\bm{x})}$. Then, since the channel is memoryless, we can write
\begin{equation}
\myPsi_{\bm{x}}=\mypsi_{x_1}\otimes\mypsi_{x_2}\otimes\cdots\mypsi_{x_n}
\label{eq:defPsi}
\end{equation}
where $\otimes$ is the Kronecker product. 

%\begin{remark}
%There is a big issue below on whether we should force the scalar products between tilted vectors to be non-negative or if we allow negative values. There is evidence that the optimal tilted vectors always have non-negative scalar products. If we force non-negative scalar products, the proof that $\vartheta(\rho)$ is an upper bound to the expurgated coefficient $E_x^{(n)}(\rho)$ for all $n$ is simpler. In that case, however, it is not simple to understand how we prove that $\vartheta(\rho)\to\vartheta$ as $\rho\to\infty$ since in principle we are not sure that Lov\'asz' vectors always have non-negative scalar products.
%\end{remark}

%\section{The $\vartheta(\rho)$ function}

The function $\vartheta(\rho)$ was derived in \cite{dalai-ISIT-2013b} as an extension of the Lov\'asz theta function as follows. Consider the inner products between the channel state vectors $\braket{\mypsi_x}{\mypsi_{x'}}\geq 0$. For a fixed $\rho\geq 1$, an \emph{orthonormal representation of degree $\rho$} of our channel $W$ is a set of ``tilted''   unit norm vectors $\{\tilde{\mypsi}_x\}$ in any Hilbert space such that $|\braket{\tilde{\mypsi}_x}{\tilde{\mypsi}_{x'}}|\leq (\braket{{\mypsi}_x}{{\mypsi}_{x'}})^{1/\rho}$. Call $\Gamma(\rho)$ the non-empty set of all possible such representations
\begin{equation}
\Gamma(\rho) = \left\{ \{\tilde{\mypsi}_x\} \, :\,  |\braket{\tilde{\mypsi}_x}{\tilde{\mypsi}_{x'}}|\leq (\braket{{\mypsi}_x}{{\mypsi}_{x'}})^{1/\rho}\right\}, 	\quad \rho\geq 1.
\end{equation}
The \emph{value} of an orthonormal representation is the quantity 
\begin{equation}
V(\{\tilde{\mypsi}_x\})=\min_{\myf}\max_x\log \frac{1}{|\braket{\tilde{\mypsi}_x}{\myf}|^2},
\end{equation}
where the minimum is over all unit norm vectors $\myf$. The optimal choice of the vector $\myf$ is called the \emph{handle} of the representation. The function $\vartheta(\rho)$ is defined as the minimum value over all representations of degree $\rho$, that is,
\begin{align}
\vartheta(\rho) & = \min_{\{\tilde{\mypsi}_x\} \in \Gamma(\rho)}V(\{\tilde{\mypsi}_x\}).
%\\
%& = \min_{\{\tilde{\psi}_x\} \in \Gamma(\rho) }\min_f\max_x\log \frac{1}{|\braket{\tilde{\psi}_x}{f}|^2}
\end{align}

%\begin{remark}
%Note that it is always possible to find an optimal representation with handle such that $\braket{\tilde{\mypsi}_x}{\myf}>0$, $\forall x$, since changing any $\tilde{\mypsi}_x$ with $-\tilde{\mypsi}_x$ gives a valid representation with the same value. In particular, we then have $\braket{\tilde{\mypsi}_x}{\myf}\geq e^{-\vartheta(\rho)/2}$ for all $x$.
%\end{remark}

The function $\vartheta(\rho)$ was used in \cite{dalai-ISIT-2013b} to derive (a stronger form of) the following bound.
\begin{theorem}
For any code of block-length $n$ with $M$ codewords and any $\rho\geq 1$, we have
\begin{equation*}
\max_{m'\neq m}\braket{\bm{\psi}_m}{\bm{\psi}_{m'}}\geq \left(\frac{ M e^{-n\vartheta(\rho)} -1}{M-1}\right)^\rho.
\end{equation*}
\label{th:genplotkin}
\end{theorem}
We observe that the proof of Theorem \ref{th:genplotkin} is essentially based on the following result that we prove here for convenience.
\begin{lemma}
Let $v_1,\ldots,v_M$ and $f$ be unit norm vectors such that $|\braket{v_i}{f}|^2\geq c>0$ for all $i$. Then 
\[
\max_{i\neq j} |\braket{v_i}{v_j}|\geq \frac{Mc-1}{M-1}
\]
\label{le:spherical}
\end{lemma}
\begin{proof}
Let $\Phi$ be a matrix whose $i$-th column is $v_i$. Then, direct computation shows that 
\[
\braket{f}{\Phi}\braket{\Phi}{f}\geq Mc.
\]
Since $f$ is a unit norm vector, $\lambda_{\text{max}}(\Phi \Phi^\dagger)\geq Mc$, where $\lambda_{\text{max}}$ is the largest eigenvalue. This also implies $\lambda_{\text{max}}(\Phi^\dagger\Phi)\geq Mc$. For a matrix $A$ with elements $A_{i,j}$, it is known that 
\begin{equation}
\lambda_{\max}(A)\leq \max_{i}\sum_j|A_{i,j}|.
\end{equation}
Applying this to $A=\Phi^\dagger \Phi$ we obtain
\begin{align*}
Mc & \leq \lambda_{\text{max}}(\Phi^\dagger\Phi)\\
& \leq \max_{i}\sum_j |\braket{v_i}{v_j}|\\
& \leq 1+(M-1)\max_{i\neq j} |\braket{v_i}{v_j}|
\end{align*}
which implies the statement of the lemma.
\end{proof}

Theorem \ref{th:genplotkin} now follows by observing that, setting $\myF=\myf^{\otimes n}$, for any sequence $ \bm{x}=(x_1\ldots,x_n)$ we have 
\begin{align}
|\braket{\tilde{\myPsi}_{\bm{x}}}{\myF}|^2 & =\prod_{i=1}^n | \braket{\tilde{\mypsi}_{x_i}}{\myf} |^2\\ 
& \geq e^{-n\vartheta(\rho)}
\label{eq:comp_inner}
\end{align}
and, for any two sequences $\bm{x},\bm{x}'$, $\braket{\myPsi_{\bm{x}}}{\myPsi_{\bm{x}'}}\geq |\braket{\tilde{\myPsi}_{\bm{x}}}{\tilde{\myPsi}_{\bm{x}'}}|^\rho$.

\section{Extension of the Bound}
\subsection{Constant Composition Codes}
The first step that we need to consider, for the development of a bound along the Elias scheme, is the extension of the umbrella bound to codes with a constant composition.
For the bound derived in the previous section, the main property of the function $\vartheta(\rho)$ that we used is the property expressed in \eqref{eq:comp_inner}.
There we see the reason for the definition of $\vartheta(\rho)$. We built a set of vectors $\{\tilde{\mypsi}_{x}\}$ associated to symbols, and a vector $\myf$ such that $\myf$ is ``close'' to all possible $\tilde{\mypsi}_{x}$. If we are interested in sequences with a particular composition, however, it can be preferable to pick $\myf$ so that $|\braket{\tilde{\mypsi}_x}{\myf}|$ is larger for the symbols $x$ which are used more frequently. This leads to a variation of $\vartheta(\rho)$ which is the analogue of the variation of the Lov\'asz theta function introduced by Marton in \cite{marton-1993}.

For a distribution $P$ and for $\rho\geq 1$, we define
\begin{equation}
\vartheta(\rho,P)=\min_{ \{\tilde{\psi}_x\} \in \Gamma(\rho), f}\sum_x P(x)\log\frac{1}{|\braket{\tilde{\mypsi}_x}{\myf}|^2}.
\end{equation}
%Again we observe, as we did for $\vartheta(\rho)$, that it is not restrictive to assume that $\braket{\tilde{\mypsi}_x}{\myf}> 0$, $\forall x$.

With this definition, if $\bm{x}$ is a sequence with composition $P$, and $\{\tilde{\mypsi}_x\}$ is a representation with handle $f$ achieving $\vartheta(\rho,P)$, we have
\begin{eqnarray}
|\braket{\tilde{\myPsi}_{\bm{x}}}{\myF} |^2 & = & \prod_{i=1}^n | \braket{\tilde{\mypsi}_{x_i}}{\myf}|^2\\
& =& \prod_{x}|\braket{\tilde{\mypsi}_{x}}{\myf}|^{2 n P(x)}\\
& =& e^{n\sum_x P(x)\log |\braket{\tilde{\mypsi}_{x}}{\myf}|^2}\\
& = & e^{-n\vartheta(\rho,P)}.
\end{eqnarray}
Then, assume we have a code with $M$ codewords $\bm{x}_1,\ldots,\bm{x}_M$ of composition $P$. We can apply Lemma \ref{le:spherical} to the vectors $\tilde{\myPsi}_{\bm{x}_i}$ and then the inequality $\braket{\myPsi_{\bm{x}}}{\myPsi_{\bm{x}'}}\geq |\braket{\tilde{\myPsi}_{\bm{x}}}{\tilde{\myPsi}_{\bm{x}'}}|^\rho$ to deduce that
\begin{align}
\max_{m\neq m'}\braket{\myPsi_{\bm{x}_m}}{\myPsi_{{\bm{x}}_{m'}}} & \geq \left( \frac{M e^{-n\vartheta(\rho,P)}-1}{M-1} \right)^\rho\\
& \geq \left( e^{-n\vartheta(\rho,P)}-M^{-1} \right)^\rho.
\end{align}
Now, we see that if $R>\vartheta(\rho,P)$, as $n\to\infty$ the above quantity goes to zero as $e^{-n\rho \vartheta(\rho,P)}$. 

Define then the asymptotic minimum distance
\begin{equation}
d(R,P)=\limsup_{n\to \infty}\max_{\mathcal{C}}\left[-\frac{1}{n}\log \max_{m\neq m'} \braket{{\myPsi}_m}{{\myPsi}_{m'}}\right]
\end{equation}
where the maximum is over all codes of length $n$, rate at least $R$ and compositions tending to $P$ as $n\to\infty$.
We have the following result.

\begin{theorem}
For any $\rho\geq 1$, if $R>\vartheta(\rho,P)$, then $d(R,P)\leq \rho\vartheta(\rho,P)$.
\label{th:const_comp_1}
\end{theorem}

It is obvious from the definitions that $\vartheta(\rho,P)\leq \vartheta(\rho)$. Hence, even after optimization of the distribution $P$, the bound derived here is at least as good as the one that we can derive from Theorem \ref{th:genplotkin}. The variation introduced here is however also useful in the case of cost constraints.

\subsection{The Elias Bound}

We now extend further the definition of $\vartheta$ in order to apply the scheme developed by Blahut as a generalization of the Elias bound. What we need now is to extend the definition of $\vartheta(\rho,P)$ to deal with stochastic matrices.
Given a distribution $P$ and a $|\mathcal{X}|\times |\mathcal{X}|$ stochastic matrix $V(x'|x)$, we define
\begin{align}
\vartheta(\rho,P,V) & = \sum_x P(x) \vartheta(\rho, V(\cdot|x))\label{eq:defthetarhoPV}\\
& =\min\sum_{x,x'} P(x)V(x'|x)\log\frac{1}{|\braket{\tilde{\mypsi}_{x,x'}}{\myf_x}|^2}
\end{align}
where the minimum is over all \emph{sequences} of representations $\{\mypsi_{x, 1},\ldots,\mypsi_{x,|\mathcal{X}|}\}\in\Gamma(\rho)$, $x\in\mathcal{X}$ (one representation for each $x$) and over all sets of unit norm vectors $\{f_x\}$, $x\in\mathcal{X}$ (a different handle for each $x$). 

Consider now the set of optimal representations and optimal handles which achieve $\vartheta(\rho,P,V)$. Let $\bm{x}=(x_1,x_2,\ldots,x_n)$ be a sequence with composition $P$ and define
\begin{equation}
\bm{f}=f_{x_1}\otimes f_{x_2}\cdots\otimes f_{x_n}
\end{equation}
For a sequence $x'=(x'_1,x'_2,\ldots,x'_n)$ which has a conditional composition $V$ given the sequence $x$, consider the vector
\begin{equation}
\tilde{\myPsi}_{\bm{x'}}=\tilde{\mypsi}_{x_1, x'_1}\otimes \tilde{\mypsi}_{x_2, x'_2}\cdots\otimes \tilde{\mypsi}_{x_n, x_n'}
\end{equation}
Then, we have
\begin{eqnarray}
|\braket{\tilde{\myPsi}_{\bm{x}'}}{\myF} |^2 & = & \prod_{i=1}^n | \braket{\tilde{\mypsi}_{x_i, x'_i}}{\myf_{x_i}}|^2\\
& =& \prod_{x,x'}|\braket{\tilde{\mypsi}_{x,x'}}{\myf_{x}}|^{2 n P(x)V(x'|x)}\\
& =& e^{n\sum_{x,x'} P(x)V(x'|x)\log|\braket{\tilde{\mypsi}_{x,x'}}{\myf_x}|^2}\\
& = & e^{-n\vartheta(\rho,P,V)}.
\end{eqnarray}
Proceeding as we did in our previous bounds, if we have a set of $M$ codewords all with a conditional composition $V$ from a fixed sequence $\bm{x}$ with composition $P$, then 
\begin{equation}
\max_{m\neq m'}\braket{\myPsi_{\bm{x}_m}}{\myPsi_{{\bm{x}}_{m'}}}\geq \left( \frac{M e^{-n\vartheta(\rho,P,V)}-1}{M-1} \right)^\rho.
\label{eq:EBD_bound}
\end{equation}

In order to use this inequality for a given code, it is now necessary to consider the possible joint compositions of a subset of codewords with some given fixed sequence $\bar{\bm{x}}$. 
Given a code with $M=e^{nR}$ codewords of composition $P$, for a $\rho\geq 1$ and $\varepsilon>0$, assume that there exists a stochastic matrix $V(x'|x)$ such that $nP(x)V(x'|x)$ is an integer,
\begin{equation}
\sum_{x}P(x)V(x'|x)=P(x')
\end{equation}
(that we will write as $PV=P$), and
\begin{equation}
R\geq I(P,V)+\vartheta(\rho,P,V)+\varepsilon,
\end{equation}
where $I(P,V)$ is the mutual information with the notation of \cite{csiszar-korner-book}.
Then, (see \cite{blahut-1977}, proof of Th. 8) there is at least one sequence $\bar{\bm{x}}$ of composition $P$ (not necessarily a codeword) such that there are at least $T=e^{n(\vartheta(\rho,P,V)+\varepsilon-o(1))}$ codewords with conditional composition $V$ from $\bar{\bm{x}}$.
Let $\mathcal{T}$ be the set of such codewords, which plays the same role as in Section \ref{sec:Intro}. Then, for these codewords we can apply the bound of equation \eqref{eq:EBD_bound} with $T$ in place of $M$. Considering the first order exponent, we then deduce that
\begin{equation}
-\frac{1}{n}\log \max_{m\neq m'} \braket{{\myPsi}_m}{{\myPsi}_{m'}}\leq \rho \vartheta(\rho,P,V) + o(1).
\end{equation}
For fixed $n$, the choice of $V$  is  constrained to satisfy the usual type constraints, but asymptotically as $n\to\infty$ this constraints can be neglected.
As a consequence, we have the following theorem.
\begin{theorem}
For given $R$, $P$ and $\rho\geq 1$, let $V$ be a $|\mathcal{X}|\times |\mathcal{X}|$ stochastic matrix such that $PV=P$. If $R > I(P,V)+\vartheta(\rho,P,V)$, then  $d(R,P)\leq§ \rho \vartheta(\rho,P,V)$.
\label{th:const_comp_2}
\end{theorem}

\begin{remark}
We observe that with the choice $V(x'|x)=P(x')$ we have $PV=P$, $I(P,V)=0$ and $\vartheta(\rho,P,V)=\vartheta(\rho,P)$. Hence, if $R>\vartheta(\rho,P)$ for a given $\rho$, the particular choice $V(x'|x)=P(x')$ gives the same bound of Theorem \ref{th:const_comp_1}, which is thus included as a particular case in Theorem \ref{th:const_comp_2}.
\end{remark}

\section{An Analysis of the Bound}

\subsection{Binary Channels}
Since most readers are probably familiar with the original Elias bound, we first give evidence that the proposed bound is a generalization by showing in detail how the original bound for binary channels is recovered as a special case. This shows that, even in the binary case, there is no loss in the use of equation \eqref{eq:exp_plotkin} with the approach based on $\vartheta$ with respect to the standard use of the Plotkin bound \eqref{eq:plotkin} under composition constraints.
In particular, the original bound for binary channels is obtained in the limit $\rho\to\infty$.

For a binary channel, let $Z=-\log\psi_0^\dagger \psi_1$ be the Battacharyya distance between the two inputs ($0$ and $1$). Then, for any $\rho$
it is not difficult to see that one can always take as an optimal representation of degree $\rho$ the two-dimensional vectors 
\begin{align*}
\tilde{\psi}_0& =[\cos(\alpha), \sin(\alpha)]^\dagger\\
\tilde{\psi}_1& =[\cos(\alpha), -\sin(\alpha)]^\dagger
\end{align*}
where $\alpha$ satisfies $\cos(2\alpha)=e^{-Z/\rho}$. For a given distribution $Q$, let the optimal handle which achieves $\vartheta(\rho,Q)$ be
\begin{align*}
f& =[\cos(\beta), \sin(\beta)]^\dagger.
\end{align*}
Then 
\begin{equation}
\vartheta(\rho,Q)= - 2Q(0)\log \cos(\alpha-\beta)- 2Q(1)\log \cos(\alpha+\beta).
\label{eq:binarytheta}
\end{equation}
where the value of $\beta$ can be determined by minimizing this expression. Upon differentiation and a little of algebra  we find
\begin{equation}
%P(0)\left(\sin(2\alpha)-\sin(2\beta)\right)=P(1)\left(\sin(2\alpha)+\sin(2\beta)\right)
\sin(2\beta)=(Q(0)-Q(1))\sin(2\alpha).
\label{eq:alphabeta}
\end{equation}
The value of $\vartheta(\rho,Q)$ can now be computed analytically by using this relation in \eqref{eq:binarytheta}. The resulting expression is complicated and not 
very useful here. So, we only study the bound of Theorem \ref{th:const_comp_2} asymptotically obtained by letting $\rho\to\infty$. We also only study the bound obtained for the uniform composition $P$, since we already know that this is the interesting case for the original Elias bound.

First note that, for any $V$, $\vartheta(\rho,P,V)\to 0$ as $\rho\to\infty$, which means that we can obtain a bound for any $R$ by choosing $V$ such that $I(P,V)<R$.
Let us then choose $V$ such that $V(1|0)=V(0|1)=\lambda$, with $\lambda$ such that  $I(P,V)=1-h(\lambda)<R$, where $h(\cdot)$ is the binary entropy function.
If we set $Q=V(\cdot|0)$, then by symmetry we have $\vartheta(\rho,P,V)=\vartheta(\rho,Q)$.
In the limit $\rho\to\infty$, since  $\cos(2\alpha)=e^{-Z/\rho}$, we have $\alpha\to 0$, and from equation \eqref{eq:alphabeta} we deduce that
$\beta\approx \alpha(1-2\lambda)$. The expression for $\vartheta(\rho,Q)$ is then asymptotically
\begin{align*}
\vartheta(\rho,Q) & \approx  - 2(1-\lambda)\log \cos(2\lambda\alpha)- 2\lambda\log \cos(2(1-\lambda)\alpha)\\
& \approx (1-\lambda)(4\lambda^2\alpha^2)+\lambda(4(1-\lambda)^2\alpha^2)\\
& = 4\lambda(1-\lambda)\alpha^2.
\end{align*}
Using again the relation $e^{-Z/\rho}=\cos(2\alpha)$ we deduce that
\begin{align}
\rho & = \frac{-Z}{\log\cos(2\alpha)}\\
&\approx  \frac{Z}{2\alpha^2}.
\end{align}
So, $\rho\vartheta(\rho,Q)\approx 2\lambda(1-\lambda)Z$. The bound of Theorem \ref{th:const_comp_2} states that for $R>\vartheta(\rho,P,V)+I(P,V)$ we have $d(R,P)\leq \rho\vartheta(\rho,P,V)$. Since here $\vartheta(\rho,P,V)=\vartheta(\rho,Q)\to 0$ as $\rho\to\infty$, in this limit the theorem says that if $R>1-h(\lambda)$ 
then $d(R)\leq 2\lambda(1-\lambda) Z$. This is an equivalent formulation of the Elias bound. One may wonder whether for finite $\rho$ a better bound can be obtained. Unfortunately, a rigorous analysis seems to be painful, but numerical evaluation shows that this is not the case, the optimal bound is  achieved as  $\rho\to\infty$.

\subsection{Non-Negative Definite Channels and Euclidean Space Codes}

The detailed analysis of the bound obtained for the BSC as $\rho\to\infty$ can be extended to all non-negative definite channels without a zero-error capacity. In this case, the bound obtained as $\rho\to\infty$ is precisely the same as that of Blahut. Due to space limitation, we can only give a sketch of the proof. For a fixed value of $x$, consider the quantity $\vartheta(\rho,V(\cdot|x))$ which appears in the definition \eqref{eq:defthetarhoPV}. Let for ease of notation $Q=V(\cdot|x)$, so that we can focus on the evaluation of  $\vartheta(\rho,Q)$ for a general $Q$ and get rid of $x$.
%The evaluation of $\vartheta(\rho,Q)$ involves the choice of vectors $\tilde{\psi}_{x}$ and handle $f$.
As mentioned in \cite{dalai-ISIT-2013b}, for these channels, for any $\rho\geq 1$, representations of degree $\rho$ exist which meet the constraints $\braket{\tilde{\psi}_{x_1}}{\tilde{\psi}_{x_2}}\leq \braket{{\psi}_{x_1}}{{\psi}_{x2}}^{1/\rho}$ with equality. All these vectors  tend to concentrate in a small cap on the unit sphere as $\rho\to\infty$, and $\vartheta(\rho,Q)\to 0$. Using the asymptotic expansions $\sin(2t)\approx 2t$ and $\log(\cos^2(t))\approx -t^2$, valid for $t\to 0$, one finds that the optimal choice of the handle $f$ for achieving $\vartheta(\rho,Q)$ is asymptotically the center of mass of the points $\tilde{\psi}_{x}$ (if vectors are interpreted as points and $Q$ as a mass distribution). 
Then one comes to the conclusion that, as $\rho\to\infty$, 
\begin{equation}
\rho\vartheta(\rho,Q)\to -\sum_{x_1,x_2}Q(x_1)Q(x_2)\log \braket{{\psi}_{x_1}}{{\psi}_{x_2}}.
\end{equation}
Furthermore, since $\vartheta(\rho,Q)\to 0$, the constraint on the rate becomes $R>I(P,V)$, and the bound on the distance takes the same form as Blahut's one. So, our bound is actually a generalization of Blahut's to general channels possibly with a zero-error capacity.

Finally, we point out that the bound derived by Piret for the squared euclidean distance of codes on the unit circle can also be obtained as a particular case of our bound. This is however due to a rather interesting independent fact, namely that for any choice of points in a euclidean space there exists a set of unit norm vectors in some other space whose pairwise Bhattachryya distances are precisely the squared euclidean distances between the original points. These vectors trivially satisfy the properties required for non-negative definite channels and thus Blahut's bound applies. In fact, Piret's bound can be recast as a special case of Blahut's one and it is thus also included in our method.

\subsection{Complexity}
The proposed bound has a non-trivial complexity and the evaluation of the optimal choice of $\rho$ and $V$ for a given channel, given $P$ and $R$ does not seem to be simple. It must be stressed, however, that for any choice of $\rho$ and $V$ we obtain a bound for the rate $R=\vartheta(\rho,P,V)+I(P,V)$. There are two main factors that should be analyzed for a deeper understanding of whether such a high complexity is reasonable or not for these kind of bounds.
One reason is that we have no closed form expression for the function $\vartheta(\rho,P,V)$, and this prevents any particularly interesting analytic study of the resulting bound. This is due to the fact that we want to cope with channels with a zero-error capacity and that we chose to build upon the work of Lov\'asz, since it is the most effective in this sense, which also does not lead to closed form expressions for bounds to $C_0$. 

Another source of complexity, instead, seems to be intrinsic in all attempts to generalize the Elias bound. We should spend a few words on this. In \cite{berlekamp-book-1984}, Berlekamp uses a slightly different approach when compared to \cite{blahut-1977} and \cite{piret-1986} and, since he considers only the case of the Hamming and Lee metrics, he uses a symmetry argument to derive a bound which has a simpler form. Still, it can be checked that the bound does not give a closed form relation between $R$ and the minimum distance, since it involves the inversion of \cite[eq. (13.15)]{berlekamp-book-1984}, which requires a non trivial computation. Furthermore, in the general case where there is no symmetry in the distances, this method cannot be applied.

The bound proposed by Piret in \cite{piret-1986}, as well, is only valid for a certain symmetric setting but, as mentioned before, it is much related to Blahut's. Using the symmetry, he obtains a bound that can be expressed in terms of one distribution (his $\beta$, in place of our stochastic matrix $V$), but even there, there is no closed form expression for the optimal distribution to choose for a given rate $R$, although they conjecture what it could be (see discussion after  \cite[eq. (41)]{piret-1986}).

Finally, the only approach which works for DMCs is the one proposed by Blahut. He gives a complete description of his bound $E_U(R)$ in a form (see his Definition 3) which looks very similar to classical bounds to the reliability function\footnote{We focus on his function $E_U(R)$ only as an upper bound on the minimum Bhattacharyya distance in this paper. The use of $E_U(R)$ as an upper bound to $E(R)$ seems to be valid only for pairwise reversible channels.} $E(R)$. Unfortunately, however, the computation of this function is in our opinion much more difficult than expected. The problem is that there seems to be a sign error in the proof of Lemma 5 which leads to erroneously consider his function $F(P)$ convex while it is actually concave. So, in our opinion, Theorem 7 is not  valid, and  in the definition of $E_U(R)$ we have a \emph{minimization} of a concave function over a convex set. Thus, the evaluation of $E_U(R)$ it much more difficult than expected.

We close by pointing out that this concavity issue of Blahut's $F(P)$ mentioned above is essentially the same reason which prevents a closed form expression for Piret's bound. It is actually even present in the original case of binary channels, although with a trivial solution; it is the point where, for a give $R$, we need to find infimum value of the concave function $2\lambda(1-\lambda)$ under the constraint that $1-h(\lambda)<R$. On the other hand, this concavity is also a key property which is needed in the original Elias bound and in all above mentioned extensions, see \cite[after (2.46)]{shannon-gallager-berlekamp-1967-2}, \cite[after (13.48)]{berlekamp-book-1984}, \cite[Lemma 6]{blahut-1977} and \cite[Lemma 4.2]{piret-1986}. Hence, we doubt that a simpler solution could ever been found along these lines. It may just be more effective to find an empirically good selection of $V$ as suggested by Piret for his $\beta$.

%

%% References:
%% We recommend the usage of BibTeX:
%%
%\bibliographystyle{IEEEtran}
%\bibliography{bibeit}


%%
%% where we here have assume the existence of the files
%% definitions.bib and bibliofile.bib.
%% BibTeX documentation can be obtained at:
%% http://www.ctan.org/tex-archive/biblio/bibtex/contrib/doc/
%%
%%
%%
%% Or manual references (pay attention to consistency!):

\end{document}